\documentclass[preprint,12pt]{elsarticle}
\usepackage[margin=1.05in]{geometry}

\usepackage{graphicx}      % include this line if your document contains figures
\usepackage{natbib}        % required for bibliography
\usepackage[english]{babel}
\usepackage{amsthm}
\usepackage{amsmath}
\usepackage{amssymb}
\usepackage{enumitem}
\usepackage{booktabs}
\usepackage{float}
\usepackage{subcaption}
\usepackage{tikz}
\usepackage{mathtools}
\usetikzlibrary{positioning, graphs,quotes}
\usepackage{booktabs}
\usepackage{tabularx}
\allowdisplaybreaks[2]
\newtheorem{theorem}{Theorem}%[section]
\newtheorem{lemma}{Lemma}%[section]
\newtheorem{remark}{Remark}%[section]
\newtheorem{assumption}{Assumption}

\begin{document}

\begin{frontmatter}

\title{\textbf{Event-Triggered Pinning Impulsive Control of Complex Networks with Actuation Delays: Stability Analysis and Zeno-Free Conditions}\tnoteref{label0}}

\tnotetext[label0]{This work was supported by the Natural Sciences and Engineering Research Council of Canada (NSERC). K. Zhang further acknowledges support from the Verner Family Faculty Fellowship in Smith Engineering at Queen's University.}

\author[a]{Ethan Astri}\ead{22kc41@queensu.ca}
\author[b]{Hamza Saleem}\ead{24sm13@queensu.ca}
\author[a,b]{Kexue Zhang\corref{cor1}}\ead{kexue.zhang@queensu.ca}
\cortext[cor1]{Corresponding author.} 

\address[a]{Department of Mathematics and Statistics, Queen's University, Canada}
\address[b]{Department of Electrical and Computer Engineering, Queen's University, Canada}

\begin{abstract}                % Abstract of 50--100 words

This paper investigates the stabilization of complex networks via event-triggered pinning impulsive control in the presence of actuation delays. Unlike existing event-triggered impulsive control schemes that assume instantaneous implementation, the proposed framework explicitly accounts for the delay between event detection and impulse execution. By constructing suitable Lyapunov functions and analyzing the network dynamics during the delay intervals, explicit delay-dependent sufficient conditions are derived to guarantee asymptotic stability. The obtained conditions characterize the interplay among network topology, actuation delays, impulsive control gains, and triggering parameters. In addition, a strictly positive lower bound on inter-event times is established, which excludes Zeno behavior and ensures practical implementability. A topology-based criterion for selecting pinned nodes is also developed through a spectral condition on the network Laplacian. Numerical simulations on a network of coupled Chua circuits illustrate the design procedure and verify the effectiveness of the proposed method.

\end{abstract}

% 5 to 10
\end{frontmatter}

\section{Introduction}
Complex networks have attracted significant attention in recent decades due to their broad applications in engineering, physics, biology, communication systems, and social sciences. Many real-world interconnected systems, including power grids, neural networks, multi-agent systems, sensor networks, and cyber-physical systems, can be modeled within the framework of complex networks (see, e.g., \cite{StructCNs, SyncCDNS}). Because of the strong coupling and nonlinear interactions among network nodes, the analysis and control of collective dynamical behaviors such as synchronization, consensus, and stabilization have become important research topics in nonlinear science and control theory (see, e.g., \cite{TANG2014184}).

Among the many available control approaches, impulsive control has proven effective and economical for controlling complex networks (see, e.g., \cite{SyncCDNS,TANG2014184,Kexue}). Unlike continuous feedback control, impulsive control modifies the system state only at discrete time instants, thereby significantly reducing communication, computation, and actuation costs (see \cite{liu2019}). In practical applications, impulsive control is particularly attractive for large-scale networked systems, where continuous communication among sensors, controllers, and actuators may be difficult or expensive to maintain. Existing impulsive control strategies can generally be classified as time-triggered or event-triggered. In time-triggered impulsive control, impulses are applied periodically or according to a prescribed schedule. Although such methods are relatively simple to implement, they may introduce unnecessary control updates when the system's state evolves near the desired equilibrium. Event-triggered impulsive control addresses this issue by generating impulses only when a certain triggering condition is violated, thereby improving communication efficiency and reducing control frequency (see, e.g., \cite{LI2020,StabStoc,SyncAnal,Quasi}).

Another important issue in controlling complex networks is the high dimensionality of large-scale systems. In many applications, directly controlling every node in a network is impractical or prohibitively expensive. Pinning control has emerged as an efficient approach to address this challenge by applying control inputs to only a subset of network nodes while utilizing the network coupling structure to propagate stabilization or synchronization throughout the entire network (see, e.g., \cite{PinCDNs,DELELLIS}). {{Combining pinning strategies with impulsive and event-triggered mechanisms offers a promising framework for efficient control of large-scale nonlinear networks at reduced control cost.}}

Recently, event-triggered impulsive pinning control has received increasing attention. Existing work has established various synchronization and stabilization criteria for complex networks under impulsive and event-triggered control mechanisms (see, e.g., \cite{PIS,1ImpCtrl,SHEN2022,LIN2022}). However, most available results assume that control impulses are implemented instantaneously upon satisfying a triggering condition. In practical networked and cyber–physical systems, this assumption is often unrealistic due to communication latency, actuator dynamics, signal transmission delays, packet scheduling, and computational overhead. As a result, there may exist a non-negligible delay between the triggering instant and the actual execution of the impulsive control input. Such actuation delays may significantly affect network dynamics and, if not properly accounted for, can destroy stability.

Compared with delays arising from system dynamics or coupling channels, actuation delays pose additional analytical challenges because the system continues to evolve after an event is triggered but before the impulsive control is applied. As a result, the mismatch between the triggering instant and the control execution instant may allow the network state to evolve significantly during the delay interval, potentially resulting in transient growth and substantially affecting the triggering behavior. Unlike standard delayed systems, the network state keeps evolving during the interval between event detection and impulsive implementation, leading to transient growth phenomena and complicated interactions among network coupling, triggering thresholds, and delayed impulsive corrections. These features make the stability analysis much more difficult and require new delay-dependent estimation techniques. Although impulsive control and event-triggered control with delays have been investigated in several settings (see, e.g., \cite{XL-JC-XL-MR:2023,CY-RG-JC-XY:2023,JC-JY-ZY-CJ:2026,DL-XX-VS:2026}), the problem of event-triggered pinning impulsive control for complex networks with \textit{actuation delays} remains largely open. 

Motivated by these observations, this paper investigates the stabilization problem for complex networks under event-triggered pinning impulsive control with actuation delays. Different from existing studies that neglect implementation delays, we explicitly incorporate the delay between event detection and impulsive execution into the control framework. By carefully analyzing the system's evolution during delay intervals, we derive explicit delay-dependent sufficient conditions guaranteeing asymptotic stability of the network. In addition, a strictly positive lower bound on inter-event times is established to exclude Zeno behavior and ensure practical implementability of the proposed triggering mechanism. Compared with existing delay-free event-triggered impulsive control frameworks, the proposed method explicitly captures the influence of actuation delays and provides a unified analysis framework for delayed impulsive network control systems. 

The principal novelty is that the actuation delay is incorporated directly into the event-triggered impulsive mechanism rather than being modeled as an ordinary system delay. Actuation delay fundamentally changes the hybrid dynamics because event generation and control execution are separated. This requires explicitly estimating the state evolution during the interval between event detection and impulse execution, which cannot be handled using existing delay-system techniques. A detailed comparison with several recent and closely related works is provided in the preliminary section, alongside the formulation of the network control model, highlighting the key differences and the additional analytical challenges addressed in this paper.

The remainder of this paper is organized as follows. Section~\ref{Sec2} introduces the problem formulation. Section~\ref{Sec3} presents the main stability results with the detailed proofs. Numerical examples are given in Section~\ref{Sec4} to demonstrate the effectiveness of the proposed approach. Finally, conclusions and future research directions are discussed in Section~\ref{Sec5}.

\textit{Notation.} Let $\mathbb{R}$ denote the set of real numbers. 
Let $\mathbb{R}^n$ and $\mathbb{R}^{n \times n}$ denote $n$-dimensional and $n \times n$-dimensional real spaces equipped with the Euclidean norm and the spectral norm respectively, both denoted by $\| \cdot \|$. 
Let $\mathbb{R}^+$ denote the non-negative real numbers and let $\mathbb{N}$ denote the non-negative integers. 
Denote by $I_n \in \mathbb{R}^{n \times n}$ the $n \times  n$ identity matrix. 
For matrices $A \in \mathbb{R}^{m \times n}$ and $B \in \mathbb{R}^{p \times q}$, let $A^{T}$ denote the transpose of $A$, and let $A \otimes B \in \mathbb{R}^{mp \times nq}$ denote their Kronecker product. 
Let $\lambda_{max}(\cdot)$ denote the maximum eigenvalue of a symmetric matrix. 
Let $\delta(\cdot)$ denote the Dirac delta function.

\section{Preliminaries}\label{Sec2}
Consider a complex network with $N$ identical nodes
\begin{equation} \label{eq:basic-network}
	\dot{x}_i(t) = f(t, x_i(t)) + c \sum_{j=1}^{N} a_{ij}\Gamma (x_j(t) - x_i(t)),
\end{equation} % Not sure if we really need this or can just skip to (2)
where $x_i \in \mathbb{R}^n$ is the state of the $i$th node, $f: \mathbb{R}^+ \times \mathbb{R}^n \rightarrow \mathbb{R}^n$ is continuously differentiable and $f(t,0)=0$ for all $t\in\mathbb{R}^+$, constant $c > 0$ is called the coupling strength, and $\Gamma$ is positive definite and is called the inner coupling matrix. 
The symmetric matrix $A = (a_{ij})_{N \times N}$ represents the network topology and has off diagonal entries $a_{ij}$ positive if there is an edge between nodes $i$ and $j$ and zero otherwise. 
Diagonal entries satisfy $a_{ii} = -\sum_{j=1, j \ne i}^{N} a_{ij}$, thus the sum of the entries in each row is zero.
Then network \eqref{eq:basic-network} is equivalent to 
\begin{equation}
	\dot{x}_i(t) = f(t, x_i(t)) + c \sum_{j=1}^{N} a_{ij} \Gamma x_j(t).
\end{equation}

We next make the following assumptions on $f$.

\begin{assumption}\label{as:lipschitz}
	There is a constant $L \in \mathbb{R}$ such that for all $x,y \in \mathbb{R}^n$, 
    $$\|f(t,x) - f(t,y)\| \le L \| x-y \|.$$
\end{assumption}
\begin{assumption} \label{as:matrix}
	There exists a constant matrix $K$ such that for all $x,y \in \mathbb{R}^n$, 
    $$(x-y)^T(f(t,x) - f(t,y)) \le (x-y)^TK\Gamma(x-y).$$
\end{assumption}

The goal of this work is to stabilize the network by applying control impulses to only some nodes in the network. 
We assume $l$ nodes receive impulses and the other $N-l$ nodes do not. 
We can relabel the nodes in the network so that nodes 1 to $l$ receive impulses. 
Then we can write the dynamics of the controlled network as follows
\begin{align} 
	\dot{x}_i(t) &= f(t, x_i(t)) + c \sum_{j=1}^{N} a_{ij} \Gamma x_j(t) + u_i(t), \quad i = 1, \dots, l, \nonumber \\
	\dot{x}_i(t) &= f(t, x_i(t)) + c \sum_{j=1}^{N} a_{ij} \Gamma x_j(t), \quad i = l + 1, \dots, N,
    \label{eq:network}
\end{align}
where 
\begin{equation}
	u_i(t) = -\sum_{k \in \mathbb{N}} d_ix_i(t^{(i)}_k)\delta(t - t_k^{(i)} - \tau) \quad i = 1, \dots, l,
\end{equation}
and $d_i \in \mathbb{R}$ are the impulsive control gains, $\tau$ is the actuation delay, and $\{t_k^{(i)}\}_{k \in \mathbb{N}}$ is a sequence of event times for node $i$ to be determined for updating the impulse control input. Then $\{t_k^{(i)}+\tau\}_{k \in \mathbb{N}}$ is the sequence of impulse times for executing the impulses. Between impulses, the state dynamics of each pinned node can be described as follows

\begin{align}\label{eq:error-system}
    \dot{x}_i(t) &= f(t, x_i(t)) + c \sum_{j=1}^{N} a_{ij} \Gamma x_j(t), \quad t \ne t_k^{(i)} + \tau, \cr
    \Delta x_i(t+\tau) &= -d_ix_i(t), \quad t = t_k^{(i)}, \quad i = 1, \dots, l \text{ and } k \in \mathbb{N}, \cr
    \dot{x}_i(t) &= f(t, x_i(t)) + c \sum_{j=1}^{N} a_{ij} \Gamma x_j(t), \quad i = l+1, \dots, N,
\end{align}
where $\Delta x_i(t) = x_i(t^+) - x_i(t^-)$ is the jump in the state caused by an impulse at time $t$ and $x_i(t^+)$ and $x_i(t^-)$ are right and left limits of $x_i$ at time $t$. 
We will assume that all states are left continuous so $x_i(t)=x_i(t^-)$.

Now we define a Lyapunov function $V(t) = \sum_{i=1}^{N} V_i(t)$ with $V_i(t) = \|x_i(t)\|^2$.
We construct the sequences $\{t_k^{(i)}\}_{k \in \mathbb{N}}$ to trigger the impulse control update based on the exponential decay threshold $\alpha_i \exp(-\beta(t - t_0))$ for $t \ge t_0$ where $\alpha_i$ and $\beta$ are positive constants by defining 
\begin{equation}\label{trigger}
	t_{k+1}^{(i)} = \inf \left\{ t>t_k^{(i)} + \tau : V_i(t) \ge \alpha_i\exp(-\beta(t-t_0)) \right\}, \quad i = 1, \dots, l.
\end{equation}
When $\tau=0$, the triggering condition in \eqref{trigger} reduces to the delay-free event-triggered pinning control scheme studied in~\cite{Kexue}, where convergence of the pinned-node states is guaranteed whenever $V_i(t_0)<\alpha_i$ and $d_i\in(0,1)$. However, these results do not directly extend to the case of actuation delays. Indeed, when $\tau>0$, the Lyapunov function $V_i$ may exceed the triggering threshold during the interval $(t_k^{(i)},t_k^{(i)}+\tau)$ before the corresponding impulse is applied, so convergence can no longer be guaranteed by the triggering condition alone. Moreover, although Zeno behavior was excluded in~\cite{Kexue} by proving $\lim_{k\to\infty} t_k^{(i)}=\infty$, no explicit lower bound on the inter-event times was obtained. The objective of this work is therefore to derive delay-dependent conditions ensuring asymptotic stability in the presence of actuation delays, while simultaneously establishing a positive lower bound on inter-event times.

Compared with the delay-free event-triggered pinning control framework in~\cite{Kexue}, the analysis in this paper is substantially more challenging. In the delay-free case, the impulsive control action is applied immediately once the triggering condition is satisfied, so the triggering mechanism directly regulates the Lyapunov function of each controlled node. In contrast, when actuation delays are present, the system continues to evolve between the triggering instant and the execution of the corresponding impulse. Consequently, the state may experience transient growth during the delay interval, and the stabilizing effect of the impulsive control can no longer be characterized solely by the triggering condition. Additional estimates are therefore required to bound the state evolution during the delay interval and to guarantee that the contraction induced by the delayed impulses dominates the transient growth of the network states.

\begin{remark}

The recent paper~\cite{DL-XX-VS:2026} considers neural networks with impulse delays and employs a Lyapunov-based event-triggered pinning strategy in which the pinned nodes are dynamically selected according to the largest synchronization errors. In contrast, the proposed approach uses node-wise event-triggering conditions, where the triggering condition of each pinned node depends only on its own synchronization error. Consequently, different pinned nodes may generate different impulse sequences, and not all pinned nodes need to be controlled simultaneously. Moreover, the proposed method provides a topology-based criterion for pinning-node selection. These features make the control strategy more distributed and easier to implement in large-scale networks.
\end{remark}

\section{Main Results}\label{Sec3}

This section develops the key estimates required for the stability analysis. To simplify the presentation of the subsequent lemmas and the main theorem, we first introduce the following notation and auxiliary constants.

Let
\[
W(t)=\sum_{i=l+1}^{N}V_i(t),
\]
and define
\[
\hat{\alpha}=\max_{1\le i\le l}\left\{\alpha_i\right\},
\qquad
\check{\alpha}=\min_{1\le i\le l}\left\{\alpha_i\right\}.
\]
Writing the coupling matrix \(A\) in the block form
\[
A=
\begin{bmatrix}
A_{11} & A_{12}\\
A_{21} & A_{22}
\end{bmatrix},
\]
where $A_{11}\in\mathbb{R}^{l\times l}$ and the other matrix blocks have appropriate dimensions, we introduce
\[
\Lambda=\gamma I_{N-l}+cA_{22},
\]
where \(\gamma=\|K\|\). The constants \(\eta_0,\eta_1,\eta_2\), and \(\eta_3\) defined below provide upper bounds for several intermediate estimates appearing in the main results.
\begin{align*}
    \eta_0 &= \left(\frac{2 c \|A_{12}\| \|\Gamma\| l \sqrt{\kappa\hat{\alpha}}}{2\lambda_{max}(\Lambda \otimes \Gamma) + \beta}\right)^2\cr
    \eta_1 &= l\kappa\hat{\alpha} + \eta_0 \cr
    \eta_2 &= 2 \gamma \|\Gamma\| \kappa\hat{\alpha} + 2\sqrt{2} c \|\Gamma\| \max\limits_{1 \le i \le l}\left\{|a_{ii}|\right\} \sqrt{\kappa\hat{\alpha}\eta_1}\cr
    \eta_3 &= \frac{2}{\beta\sqrt{\mu\check{\alpha}}}\left(L\sqrt{\kappa\hat{\alpha}} + \sqrt{2}c\max\limits_{1\le i\le l}\{|a_{ii}|\}\|\Gamma\| \sqrt{\eta_1}\right) 
\end{align*}
where $\mu$ and $\kappa$ are positive constants.

According to the event-triggering mechanism defined in~\eqref{trigger}, the Lyapunov function $V_i$ may exceed the threshold $\alpha_i e^{-\beta (t-t_0)}$ between two consecutive triggering instants. Therefore, the Lyapunov function associated with the pinned nodes is not necessarily constrained below the triggering threshold for all $t\ge t_0$. Nevertheless, the following lemma shows that both the network Lyapunov function $V$ and the derivative of the Lyapunov functions $V_i$ corresponding to the pinned nodes converge exponentially to zero as $t\to\infty$, provided that
\[
V_i(t)\le \kappa\alpha_i e^{-\beta (t-t_0)}
\]
for all $t\ge t_0$, where $\kappa\ge 1$ is a prescribed constant. The quantity $\kappa\alpha_i e^{-\beta (t-t_0)}$ can be regarded as an enlarged threshold obtained by scaling the triggering condition in~\eqref{trigger}.

\begin{lemma} \label{lm:bound}
Suppose $\kappa\geq 1$, $W(t_0) \le \eta_0$, and $\beta+2\lambda_{max}(\Lambda\otimes \Gamma)<0$.
	If 
    \[
    V_i(t) \le \kappa \alpha_i \exp(-\beta (t-t_0)) \textrm{ for } t \ge t_0,
    \]
    where $i = 1, \dots, l$, then 
    $$V(t) \le \eta_1 \exp(-\beta (t-t_0)),$$
	and
    $$\dot{V}_i(t) \le \eta_2 \exp(-\beta (t-t_0)) \textrm{ for } i = 1, \dots, l.$$
\end{lemma}
\begin{proof}
	We shall first find an upper bound on $W(t)$. 
	Let  $\bar{x} = (x_{l+1}^T, \dots, x_{N}^T)^T$ and $\tilde x = (x_1^T, \dots, x_{l}^T)^T$. 
	Differentiating we obtain: 
	\begin{align}\label{Wdot}
		\dot{W}(t) &= \sum_{i=l+1}^{N} 2x_i^T\left(f(t, x_i) + c \sum_{j=l+1}^{N} a_{ij} \Gamma x_j + c \sum_{j=1}^{l} a_{ij}\Gamma x_j \right) \nonumber \\
							 & \le 2\bar{x}^T \left(I_{N-l}\otimes K\Gamma + c A_{22}\otimes \Gamma\right)\bar{x}  + 2c \sum_{i=l+1}^{N} \sum_{j=1}^{l} a_{ij}x_i^T\Gamma x_j \nonumber \\
							 &\le 2\bar{x}^T(\Lambda \otimes \Gamma)\bar{x} + 2 \bar{x}^T(cA_{21} \otimes \Gamma) \tilde x \nonumber \\
							 &\le 2 \lambda_{max}(\Lambda \otimes \Gamma) W(t) + 2 \|cA_{21}\otimes \Gamma\|\;\|\bar{x}\|\;\|\tilde x\| \nonumber \\
							 &\le -\delta_1 W(t) + \delta_2 e^{-\frac{\beta}{2}(t - t_0)}\sqrt{W(t)}
	\end{align}
	Where $\delta_1 = -2\lambda_{max}(\Lambda \otimes \Gamma) > \beta$, and $\delta_2 = 2c\|A_{21}\|\|\Gamma\|l\sqrt{\kappa\hat{\alpha}}$ $> 0$. We then can conclude from~\eqref{Wdot} that if $W(t) = \eta_0 \exp(-\beta (t - t_0))$, then
    \begin{align} 
		\dot{W}(t) &\le -\delta_1W(t) + \delta_2 e^{-\frac{\beta}{2}(t - t_0)}\sqrt{W(t)} \nonumber \\
							 &= -\delta_1\eta_0e^{-\beta(t - t_0)} + \delta_2 e^{-\beta(t - t_0)}\sqrt{\eta_0} \nonumber \\ 
							 &\le \eta_0\left(-\delta_1 + \delta_2 \left(\frac{\delta_1 - \beta}{\delta_2}\right)\right) e^{-\beta(t-t_0)}\nonumber \\
							 &= -\eta_0 \beta e^{-\beta(t - t_0)} \nonumber \\
							 &= \frac{\text{d}}{\text{dt}} \left(\eta_0e^{-\beta(t-t_0)}\right), \label{eq:boundW}
	\end{align}
    which implies that if $W(\hat{t}) \leq \eta_0 \exp(-\beta (\hat{t} - t_0))$ for some $\hat{t}\geq t_0$, then $W(t) \leq \eta_0 \exp(-\beta ({t} - t_0))$ for all $t\geq \hat{t}$. Note that in~\eqref{eq:boundW} we used the definition of $\eta_0$ and the assumption $\beta+2\lambda_{max}(\Lambda\otimes \Gamma)<0$.
    Hence, since $W(t_0)\leq \eta_0$, we have
    \[
       W(t) \le \eta_0 \exp(-\beta(t-t_0)) \textrm{ for all } t \ge t_0.
    \]
	Then we can obtain an upper bound on $V(t)$, 
	\begin{equation}
		V(t) = \sum_{i=1}^{l} x_i^Tx_i + W(t) \le \left(l\kappa\hat{\alpha} + \eta_0\right) e^{-\beta (t - t_0)} = \eta_1 e^{-\beta (t - t_0)},~~~~t\geq t_0.
	\end{equation}
	Then, we can obtain an upper bound on $\dot{V}_i(t)$:
	\begin{align*}
		\dot{V}_i(t) &= 2x_i^T\left(f(t, x_i) + c \sum_{j=1}^{N} a_{ij} \Gamma x_j\right) \nonumber \\
								 &\le 2 x_i^T K \Gamma x_i + 2c\|x_i\| \|\Gamma\| \sum_{j=1}^{N} \|a_{ij}\| \|x_j\| \nonumber \\
								 &\le 2\gamma \| \Gamma \| V_i(t) + 2c \sqrt{V_i(t)} \|\Gamma \| \sqrt{\sum_{j=1}^{N} \|a_{ij}\|^2 \sum_{j=1}^{N}\|x_j\|^2 } \nonumber \\
								 &\le 2\gamma \| \Gamma \| \kappa\alpha_ie^{-\beta(t-t_0)} + 2c \sqrt{\kappa\alpha_i} e^{-\frac{\beta}{2}(t-t_0)} \|\Gamma \| \sqrt{2} |a_{ii}| \sqrt{V(t)} \nonumber \\
								 & \le \left(2\gamma \|\Gamma\|\kappa\alpha_i + 2\sqrt{2}c\sqrt{\kappa\alpha_i}\|\Gamma\||a_{ii}|\sqrt{\eta_1}\right) e^{-\beta(t - t_0)} \nonumber \\
								 &\le \eta_2 e^{-\beta(t - t_0)},
	\end{align*}
    which completes the proof.
\end{proof}

It should be noted that Lemma~\ref{lm:bound} establishes the exponential convergence of $V(t)$ to zero and therefore guarantees the attractivity of the trivial solution of network~\eqref{eq:network}. However, stability of network~\eqref{eq:network} cannot be concluded directly from Lemma~\ref{lm:bound}, since the parameter $\eta_1$ is generally nonzero even when the initial condition of the network is trivial. The following lemma resolves this issue by showing that the trivial solution of network~\eqref{eq:network} is stable under the same conditions as those of Lemma~\ref{lm:bound}.

\begin{lemma} \label{lm:stability}
	If all the conditions of Lemma~\ref{lm:bound} hold, then the network~\eqref{eq:network} is stable.
\end{lemma}
\begin{proof}
	We will obtain an increasing upper bound on $V(t)$ and a decreasing upper bound on $V(t)$ and use their intersection to get an upper bound on $V(t)$ in terms of $V(t_0)$.
	Let $\mathbf{x} = (x_1^T, \dots, x_N^T)^T$. 
	Taking the derivative along the trajectories of the system:
	\begin{align}
		\dot{V}(t) &= 2\mathbf{x}^T\dot{\mathbf{x}} \nonumber \\
							 &= \sum_{i = 1}^{N} 2x_i^T \left(f(t, x_i) + c \sum_{j=1}^{N} a_{ij}\Gamma x_j\right) \nonumber \\
							 &\le \sum_{i=1}^{N} 2x_i^T K\Gamma x_i + 2c \sum_{i=1}^{N} \sum_{j=1}^{N} a_{ij} x_i \Gamma x_j \nonumber \\
							 &\le 2\mathbf{x}^T \left(I_N \otimes (K\Gamma) + 2cA \otimes \Gamma\right) \mathbf{x} \nonumber \\
							 &\le \sigma V(t),
	\end{align}
	where $\sigma = 2\|I_N \otimes(K\Gamma) + 2cA \otimes \Gamma \|$.
	Which implies 
	\begin{equation}\label{eq:increasing-bound}
		V(t) \le h_1(t):= V(t_0)e^{\sigma (t-t_0)} \textrm{ for } t\geq t_0.
	\end{equation}
%	Without loss of generality assume $\delta > 0$.
	By Lemma \ref{lm:bound}, we have
	\begin{equation}\label{eq:decreasing-bound}
		V(t) \le h_2(t):= \eta_1e^{-\beta(t - t_0)} \textrm{ for } t\geq t_0.
	\end{equation}
	Now function $h_1$ is increasing and $h_2$ is decreasing so the following equation has a unique solution
	\begin{equation}
		V(t_0)e^{\sigma(t - t_0)} = \eta_1e^{-\beta(t - t_0)}.
	\end{equation}
	We solve for $t$ and call this value
	\begin{equation}
		t_{max} = \frac{\ln\left(\frac{\eta_1}{V(t_0)}\right)}{\beta + \sigma} + t_0.
	\end{equation}
	Then $V(t)$ is less than or equal to the value of $h_1(t)$ and $h_2(t)$ at $t_{max}$ for $t \ge t_0$. 
	We denote this value $V_{max}$.
	Substituting $t_{max}$ into $h_1(t)$ (or $h_2(t)$), we obtain 
	\begin{align}
		V(t) \le  V_{max}&:=h_1(t_{max}) \nonumber\\
        &= V(t_0)e^{\sigma (t_{max} - t_0)}  \\
										 &= V(t_0)e^{\sigma ((\frac{\ln(\eta_1V(t_0)^{-1})}{\beta + \sigma} + t_0) - t_0)} \nonumber \\
										 &= V(t_0) \left(\frac{\eta_1}{V(t_0)}\right)^{\frac{\sigma}{\beta + \sigma}} \nonumber \\
										 &= V(t_0)^{\frac{\beta}{\beta + \sigma}}\eta_1^{\frac{\sigma}{\beta + \sigma}}
	\end{align}
So $V_{max} \rightarrow 0$ as $V(t_0) \rightarrow 0$ which implies the stability of network~\eqref{eq:network}.
\end{proof}

%\begin{remark}
	When $\tau=0$, the proposed framework reduces to the delay-free event-triggered pinning control scheme studied in \cite{Kexue}. In this case, Lemma 1 recovers the synchronization result of \cite{Kexue}. However, the present analysis provides several additional results that were not obtained in \cite{Kexue}, including explicit bounds on $V(t)$ and $\dot V_i(t)$, asymptotic stability of the trivial solution, and an explicit lower bound on the inter-event times. More importantly, the techniques developed in this paper allow the analysis to be extended to systems with actuation delays, which cannot be handled by the arguments in \cite{Kexue}. 
%\end{remark}

The condition $W(t_0)\le \eta_0$ is not restrictive in practice because $\eta_0$ can be enlarged by increasing the design parameters $\alpha_i$ and $\kappa$. Consequently, for any prescribed set of initial conditions, these parameters can be chosen so that $W(t_0)\le \eta_0$ holds. Thus, this condition should be interpreted as a compatibility requirement between the initial state and the triggering parameters rather than a restriction on the network dynamics. Larger values of $\alpha_i$ and $k$ enlarge the admissible initial set but may result in larger transient deviations.

The above two lemmas establish asymptotic stability of network~(4) provided that
\[
V_i(t)\le \kappa\alpha_i e^{-\beta(t-t_0)}, \quad i=1,\ldots,l.
\]
When actuation delays are present, however, the state continues to evolve after an event is triggered but before the corresponding impulse is applied. Consequently, the delay, the triggering parameter $\kappa$, and the contraction factor $\mu$ cannot be chosen independently. The delay must be sufficiently small to ensure that the state remains below the enlarged threshold $\kappa\alpha_i e^{-\beta(t-t_0)}$ during the delay interval, while the impulsive control gains must be sufficiently large to reduce the state below the tighter bound $\mu\alpha_i e^{-\beta(t-t_0)}$ after each impulse. The following theorem formalizes these requirements and provides explicit sufficient conditions for asymptotic stability of network~(4).

\begin{theorem}\label{thm:main}
	Let constants $\mu \in (0,1)$, {$\kappa \ge 1$}, and suppose $\Lambda < 0$. Then select parameters $\alpha_i$ ($i=1,\dots,l$) and $\beta$ in~\eqref{trigger} so that
    \[
V_i(t_0) < \mu\alpha_i, \quad W(t_0) \le \eta_0,\quad \textrm{ and } \quad \beta+2\lambda_{max}(\Lambda\otimes \Gamma)<0.
    \]
	If the actuation delay $\tau$ satisfies 
    \begin{equation}\label{tau.require}
    \tau \le \min \left\{\frac{1}{\beta} \ln \left(\frac{\kappa\check{\alpha}\beta +\eta_2}{\check{\alpha}\beta + \eta_2}\right), ~ -\frac{2}{\beta}\ln\left( \frac{\eta_3}{\eta_3+1} \right)\right\},
    \end{equation}
    and the control gains satisfy 
    \begin{equation}\label{di.require}
    (\eta_3 + 1)\left(1 - \exp(-\frac{\beta}{2}\tau)\right) \le d_i \le 1,
    \end{equation}
    for $i = 1, \dots, l$, then system \eqref{eq:error-system} is asymptotically stable.
	Furthermore, the inter-event times $\{t^{(i)}_k,t^{(i)}_{k+1}\}_{k\in\mathbb{N}}$ of the $i$th pinned nodes are bounded below by $$\tau + \frac{1}{\beta}\ln \left(\frac{\alpha_i\beta+ {\eta_2}}{\mu\alpha_i\beta+ {\eta_2}}  \right),$$
    where $i=1,2,...,l$.
\end{theorem}

\begin{proof}	First we show that $V_i(t)$ is bounded by $\kappa \alpha_i  \exp(-\beta (t-t_0))$ for $i = 1, \dots, l$ and use Lemmas \ref{lm:bound} and \ref{lm:stability} to obtain the stability result.
We use an inductive argument. 

By assumption $V_i(t_0) < \mu\alpha_i$, we just need to show that if 
\[
V_i((t_k^{(i)} + \tau)^+) \le \mu\alpha_i \exp(-\beta (t_k^{(i)} + \tau - t_0))
\]
then 
\[
V_i(t) \le \kappa\alpha_i \exp(-\beta (t-t_0)) \textrm{ for }t \in (t_k^{(i)} + \tau, t_{k+1}^{(i)} + \tau],
\]
and
\[
V_i((t_{k+1}^{(i)} + \tau)^+) \le \mu\alpha_i \exp(-\beta (t_{k+1}^{(i)} + \tau - t_0)).
\]
For $t \in (t_k^{(i)} + \tau, t_{k+1}^{(i)})$, $V_i(t)$ is less than $\kappa \alpha_i  \exp(-\beta (t-t_0))$ according to the definition of event times in~\eqref{trigger}, thus it is sufficient to show 
\[
V_i(t) \le \kappa\alpha_i \exp(-\beta (t-t_0)) \textrm{ for }t \in [t_{k+1}^{(i)}, t_{k+1}^{(i)} + \tau],
\] 
and 
\[
V_i((t_{k+1}^{(i)} + \tau)^+) \le \mu \alpha_i \exp(-\beta (t_{k+1}^{(i)} + \tau - t_0)).
\]
To do this, we use a contradiction argument and suppose that there exists some $t \in [t_{k+1}^{(i)}, t_{k+1}^{(i)} + \tau]$ so that
\[
V_i(t) > \kappa\alpha_i \exp(-\beta (t-t_0)).
\]
Then let 
\[
\hat{t} = \inf \left\{t \in [t_{k+1}^{(i)}, t_{k+1}^{(i)} + \tau] : V_i(t) = \kappa\alpha_i\exp(-\beta(t-t_0))\right\},
\]
which then implies that $\hat{t}\in (t_{k+1}^{(i)}, t_{k+1}^{(i)} + \tau)$ and $V_i(t) \le \kappa\alpha_i\exp{(-\beta(t-t_0))}$ for $t \in [t_{k+1}^{(i)}, \hat{t}]$. By Lemma \ref{lm:bound}, we get 
\[
\dot{V_i}(t) \le \eta_2\exp(-\beta(t-t_0)) \textrm{ for } t \in [t_{k+1}^{(i)}, \hat{t}].
\]
Then, noting that $V_i(t_{k+1}^{(i)}) = \alpha_i\exp(-\beta(t_{k+1}^{(i)} - t_0))$, we have that $V_i(t)$ is bounded by the solution to the following initial value problem for $t \in [t_{k+1}^{(i)}, \hat{t}]$:
\begin{equation} \label{eq:ivp}
	\dot{y}(t) = \eta_2e^{-\beta{(t-t_0)}}, \quad y(t_{k+1}^{(i)}) = \alpha_ie^{-\beta(t_{k+1}^{(i)} - t_0)},
\end{equation}
which has solution $y$ satisfying
\begin{equation} \label{eq:Vibound}
	V_i(t) \le y(t) = \left(\alpha_i + \frac{\eta_2}{\beta}\right)e^{-\beta(t_{k+1}^{(i)} - t_0)} - \frac{\eta_2}{\beta}e^{-\beta(t - t_0)}.  
\end{equation}
Now $y$ is increasing, $\kappa\alpha_i\exp(-\beta(t - t_0))$ is decreasing, and $V_i(t_{k+1}^{(i)}) < \kappa\alpha_i\exp(-\beta(t_{k+1}-t_0))$ so there is a unique intersection between the graphs of $y$ and the function $\kappa\alpha_i\exp(-\beta(t-t_0))$ which occurs at 
\begin{equation}\label{yintersect}
	\tilde{t} = t_{k+1}^{(i)} + \frac{1}{\beta}\ln\left(\frac{\kappa\alpha_i\beta + \eta_2}{\alpha_i\beta + {\eta_2}}\right) \ge t_{k+1}^{(i)} + \tau.
\end{equation}
Because $y(t)$ is an upper bound on $V_i(t)$, we then can derive from~\eqref{yintersect} that $\tilde{t}\geq \hat{t} \ge t_{k+1}^{(i)} + \tau$, which is a contradiction to the definition of $\hat{t}$. Thus 
\[
V_i(t) \leq \kappa\alpha_i\exp(-\beta(t - t_0))  \textrm{ for } t \in [t_{k+1}^{(i)}, t_{k+1}^{(i)} + \tau].
\]

It remains to verify that 
$$V_i((t_{k+1}^{(i)} +\tau)^+) \le \mu\alpha_i \exp(-\beta (t_{k+1}^{(i)} + \tau - t_0)).$$ 
By applying Assumption $\ref{as:lipschitz}$ with Lemma \ref{lm:bound} and using the fact that $V_i(t) \le \kappa\alpha_i \exp(-\beta(t - t_0))$ for $t \le t_{k+1}^{(i)} + \tau$, we observe that
\begin{align*}
	\|x_i(t_{k+1}^{(i)} + \tau) - x_i(t_{k+1}^{(i)})\| &= \| \int_{t_{k+1}^{(i)}}^{t_{k+1}^{(i)} + \tau} \dot{x}_i(t) dt \| \nonumber \\
																 &= \| \int_{t_{k+1}^{(i)}}^{t_{k+1}^{(i)} + \tau} f(t, x_i(t)) + c \sum_{j=1}^{N} a_{ij} \Gamma x_j(t) dt \| \nonumber \\
																 &\le \int_{t_{k+1}^{(i)}}^{t_{k+1}^{(i)} + \tau} L \sqrt{V_i(t)} + 
                                                                 c \|\Gamma\| \sqrt{\sum_{j=1}^{N} \|a_{ij}\|^2 \sum_{j=1}^{N} \|x_j\|^2} dt \nonumber \\
																 &\le \int_{t_{k+1}^{(i)}}^{t_{k+1}^{(i)} + \tau} L\sqrt{\kappa\alpha_i}e^{-\frac{\beta}{2}(t-t_0)} + 
                                                                 c \|\Gamma\| \sqrt{2} \max\limits_{1 \le i \le l}\{\|a_{ii}\|\} \sqrt{V(t)} dt \nonumber \\
																 &\le \int_{t_{k+1}^{(i)}}^{t_{k+1}^{(i)} + \tau} (L\sqrt{\kappa\alpha_i} + 
                                                                 \sqrt{2}c \max\limits_{1 \le i \le l}\{\|a_{ii}\|\}\|\Gamma\| \sqrt{\eta_1})e^{-\frac{\beta}{2}(t-t_0)}  dt \nonumber \\
																 &\le \int_{t_{k+1}^{(i)}}^{t_{k+1}^{(i)} + \tau} \frac{\eta_3\beta\sqrt{\mu\alpha_i}}{2}e^{-\frac{\beta}{2}(t - t_0)} dt \nonumber \\
																 &\le \eta_3\sqrt{\mu\alpha_i}(e^{-\frac{\beta}{2}(t_{k+1}^{(i)} - t_0)} - e^{-\frac{\beta}{2}(t_{k+1}^{(i)} + \tau - t_0)}).
\end{align*}
Then using the above inequality and our assumption about $d_i$, we obtain
\begin{align*}
	& V_i((t_{k+1}^{(i)} + \tau)^+) \cr
    &= \|x_i(t_{k+1}^{(i)} + \tau) -d_ix_i(t_{k+1}^{(i)})\|^2 \nonumber \\
													&= \|x_i(t_{k+1}^{(i)} + \tau) - x_i(t_{k+1}^{(i)}) + (1-d_i) x_i(t_{k+1}^{(i)})  \|^2 \nonumber \\
													&\le (\|x_i(t_{k+1}^{(i)} + \tau) - x_i(t_{k+1}^{(i)})\| + (1-d_i)\|x_i(t_{k+1}^{(i)})\|)^2 \nonumber \\
													&\le \Big(\eta_3\sqrt{\mu\alpha_i} (e^{-\frac{\beta}{2}(t_{k+1}^{(i)} - t_0)} - e^{-\frac{\beta}{2}(t_{k+1}^{(i)} + \tau - t_0)}) + (1-d_i)\sqrt{\mu\alpha_i}e^{-\frac{\beta}{2}(t_{k+1}^{(i)} - t_0)} \Big)^2  \nonumber \\
													&\le \Big(\eta_3\sqrt{\mu\alpha_i} (e^{-\frac{\beta}{2}(t_{k+1}^{(i)} - t_0)} - e^{-\frac{\beta}{2}(t_{k+1}^{(i)} + \tau - t_0)}) + ((1-(\eta_3 + 1)(1 - e^{-\frac{\beta}{2}\tau}))\sqrt{\mu\alpha_i}e^{-\frac{\beta}{2}(t_{k+1}^{(i)} - t_0)}\Big)^2 \nonumber \\
													&=\mu\alpha_i e^{-\beta(t_{k+1}^{(i)} + \tau - t_0)}.
\end{align*}
So by the above induction argument, $V_i(t)$ is bounded by $\kappa \alpha_i \exp(-\beta (t-t_0))$ for $i = 1, \dots, l$.
Then Lemma \ref{lm:bound} implies $V(t) \le \eta_1 \exp(-\beta (t - t_0))$ for all $t \ge t_0$, which then implies that the limit as $t \rightarrow \infty$ of $V(t)$ is $0$. Furthermore, with this upper bound on $V_i(t)$, Lemma \ref{lm:stability} shows the system is stable. Therefore, we conclude that network~\eqref{eq:error-system} is asymptotically stable.

Finally, we establish a lower bound on the inter-event times for each pinned node. 
Obviously events occur at least $\tau$ apart, but we can also find a lower bound on the time between the $p$th control impulse and the $(p+1)$th event for all $p\in\mathbb{N}$. 

By Lemma \ref{lm:bound}, we have that $\dot{V}_i(t) \le {\eta_2} \exp(-\beta(t-t_0))$ for $i = 1, \dots, l$. Then, noting that $V_i((t_{k}^{(i)} + \tau)^+) \le \mu\alpha_i\exp(-\beta(t_{k}^{(i)}+\tau - t_0))$, we see that $V_i(t)$ is bounded by the solution to the initial value problem \eqref{eq:ivp2} for $t \in [t_{k}^{(i)} + \tau, t_{k+1}^{(i)}]$:
\begin{equation} \label{eq:ivp2}
	\dot{z} = {\eta_2} e^{-\beta(t-t_0)}, \quad z(t_{k}^{(i)}+\tau) = \mu\alpha_ie^{-\beta(t_{k}^{(i)}+\tau - t_0)}
\end{equation}
which has a solution $z$ satisfying
\begin{equation} \label{eq:Vibound2}
	V_i(t) \le z(t) := \left(\mu \alpha_i + \frac{\eta_2}{\beta}\right)e^{-\beta(t_k^{(i)} + \tau - t_0)} - \frac{\eta_2}{\beta}e^{-\beta(t - t_0)}.
\end{equation}
Now $z$ is increasing, $\alpha_i\exp(-\beta(t - t_0))$ is decreasing, and $V_i(t_k^{(i)} + \tau) \le \alpha_i\exp(-\beta(t_k+\tau -t_0))$ so there exists a unique intersection between the graphs of $z$ and the function $\alpha_i\exp(-\beta(t - t_0))$ which occurs at 
\begin{equation}
	t = t_{k}^{(i)} + \tau + \frac{1}{\beta} \ln \left(\frac{\alpha_i\beta+\eta_2}{\mu\alpha_i\beta+\eta_2}  \right).
\end{equation}
Because $z(t)$ is an upper bound on $V_i(t)$, it takes at least $\frac{1}{\beta} \ln \left(\frac{\alpha_i\beta+\eta_2}{\mu\alpha_i\beta+\eta_2}  \right)$ units of time for the Lyapunov function $V_i$ to evolve from $V_i((t_{k}^{(i)} + \tau)^+)$ to reach the triggering threshold in~\eqref{trigger} at time $t^{(i)}_{k+1}$. That is,  the minimum time between consecutive events for the pinned $i$th node is $\tau +\frac{1}{\beta} \ln \left(\frac{\alpha_i\beta+\eta_2}{\mu\alpha_i\beta+\eta_2} \right)$. Therefore, network~\eqref{eq:error-system} with the event times determined by~\eqref{trigger} is free of Zeno behavior.
\end{proof}

\begin{figure}[htpb]
	\centering
	\includegraphics[width=0.6\textwidth]{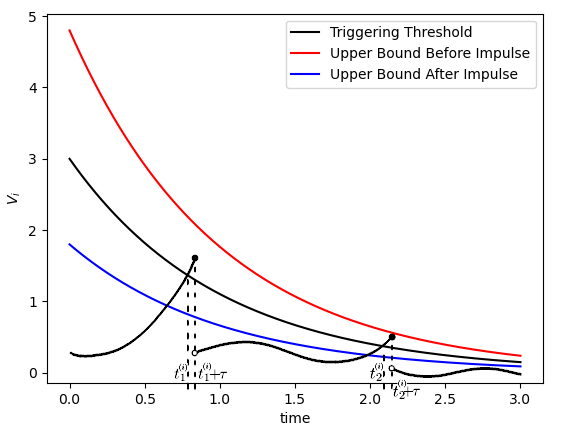}
	\caption{Lyapunov function $V_i(t)$ for the controlled node $i$ with its triggering threshold, and the scaled upper and lower bounds.}
	\label{fig:idea}
\end{figure}
The idea of the proposed control scheme is illustrated in Figure~\ref{fig:idea}. For the $i$th controlled node, we consider the triggering threshold together with two exponential bounds, one scaled by the factor $\kappa$ and the other by $\mu$. When $V_i(t)$ reaches the triggering threshold at time $t^{(i)}_k$, an impulse is scheduled and executed at time $t^{(i)}_k+\tau$ due to the actuation delay. The delay bound in Theorem~\ref{thm:main} guarantees that
\[
V_i(t)\le \kappa\alpha_i e^{-\beta(t-t_0)}
\]
throughout the interval between the triggering instant and the delayed impulse, while the control-gain condition ensures that
\[
V_i((t^{(i)}_k+\tau)^+) \le \mu\alpha_i e^{-\beta(t^{(i)}_k+\tau-t_0)},
\]
immediately after the impulse is applied. Consequently, the stability result follows from Lemmas~\ref{lm:bound} and~\ref{lm:stability}. Furthermore, Theorem~\ref{thm:main} establishes a strictly positive lower bound on inter-event times, thereby excluding Zeno behavior and ensuring practical implementability of the proposed control strategy. The parameter $\kappa\ge1$ determines the allowable overshoot before an impulse is applied, whereas $\mu\in(0,1)$ characterizes the degree of state reduction achieved by each impulse.
%\end{remark}

Theorem~\ref{thm:main} reveals that both the feasibility of asymptotic stability and the maximum admissible actuation delay depend on the selection of pinned nodes. This observation motivates a topology-based pinning strategy derived from the stability condition
\[
\Lambda=\gamma I_{N-l}+cA_{22}<0,
\]
where \(A_{22}\) is the principal submatrix of the network Laplacian corresponding to the unpinned nodes. Since the spectrum of \(\Lambda\) depends on the choice of pinned nodes, the pinning set should be selected such that
\[
\lambda_{\max}(\Lambda)<0.
\]
In practice, this can often be achieved by first pinning low-degree nodes, i.e., nodes with relatively small values of \(|a_{ii}|\), and then adding further pinned nodes until \(\Lambda\) becomes negative definite. Among all feasible pinning configurations satisfying \(\Lambda<0\), the one yielding the largest admissible delay bound in Theorem~1 is preferred. This provides a systematic topology-based guideline for selecting pinning nodes while enhancing robustness with respect to actuation delays.

\section{Numerical Example}\label{Sec4}

In this section, we demonstrate how the stability conditions developed in Theorem~\ref{thm:main} can be used to select the pinned nodes, determine the admissible actuation delay, and design the impulsive control gains. The effectiveness of the proposed approach is illustrated through a network of coupled Chua circuits.

\begin{figure}[htpb]
   \centering
    \begin{tikzpicture}[
        % Define a style for all nodes
        scale = 0.5,
        node distance=0.7cm, % Default distance between nodes
        every node/.style={circle, draw, minimum size=8mm}
    ]

    \node[circle,draw,fill=gray!25] (N1) at (0, 0) {1};
    \node[circle,draw,fill=gray!25] (N2) [above right=of N1] {2};
    \node[circle,draw,fill=gray!25] (N3) [above right=of N2] {3};
    \node (N4) [right=of N3] {4};
    \node (N5) [below right=of N4] {5};
    \node (N6) [below right=of N5] {6};
    \node (N7) [below left=of N6] {7};
    \node (N8) [below right=of N1] {8};

    % Style for edge weights (labels) - MODIFIED
    \tikzset{weight/.style={

    draw=none, fill=none, circle=false,
    font=\small}}

    \draw (N1) -- (N4) ;
    \draw (N1) -- (N5);
    \draw (N1) -- (N6) ;
    \draw (N1) -- (N7);
    \draw (N1) -- (N8);

    \draw (N2) -- (N4);
    \draw (N2) -- (N5);
    \draw (N2) -- (N6);
    \draw (N2) -- (N7);
    \draw (N2) -- (N8);

    \draw (N3) -- (N4);
    \draw (N3) -- (N5);
    \draw (N3) -- (N6);
    \draw (N3) -- (N7);
    \draw (N3) -- (N8);

    \draw (N4) -- (N1);
    \draw (N4) -- (N2);
    \draw (N4) -- (N3);
    \draw (N4) -- (N7);
    \draw (N4) -- (N8);

    \draw (N5) -- (N1);
    \draw (N5) -- (N2);
    \draw (N5) -- (N3);
    \draw (N5) -- (N7);
    \draw (N5) -- (N8);

    \draw (N6) -- (N1);
    \draw (N6) -- (N2);
    \draw (N6) -- (N3);
    \draw (N6) -- (N7);
    \draw (N6) -- (N8);

    \draw (N7) -- (N1);
    \draw (N7) -- (N2);
    \draw (N7) -- (N3);
    \draw (N7) -- (N4);
    \draw (N7) -- (N5);
    \draw (N7) -- (N6);
    \draw (N7) -- (N8);

    \draw (N8) -- (N1);
    \draw (N8) -- (N2);
    \draw (N8) -- (N3);
    \draw (N8) -- (N4);
    \draw (N8) -- (N5);
    \draw (N8) -- (N6);

    \end{tikzpicture}
    \caption{Network topology with the pinned nodes highlighted in gray. Let all the edges of nodes 1, 2, or 3 have weight 30, all edges between nodes 4, 5, or 6 and 7, or 8 have weight 35, and the edge between node 7 and 8 have weight 40.}    \label{fig:graph}
\end{figure}
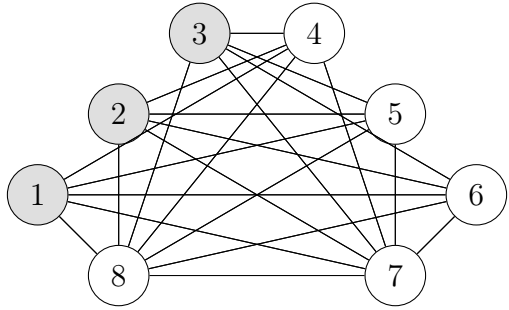

Consider the network topology shown in Figure~\ref{fig:graph} with $N=8$, and node dynamics described by the Chua Circuit from~\cite{FurPinCtrl} modeled by~\eqref{eq:chua}. The uncontrolled dynamics of this network are illustrated in Figure \ref{fig:chua}.
\begin{equation} \label{eq:chua}
	\dot{x}(t) = Bx(t) - p[h(x_1(t)),0,0]^T
\end{equation}
where $h(z) = \frac{m_0-m_1}{2}(|x+1| - |x-1|)$, \[
	B = \begin{bmatrix}-pm_1 & p & 0 \\ 1 & -1 & 1 \\ 0 & -q & 0\end{bmatrix},
\]
and $p = 2.4$, $q = 2.7$, $m_0 = -0.5$, and $m_1 = 0.2$.

\begin{figure}[htpb]
	\centering
	\includegraphics[width=0.45\textwidth]{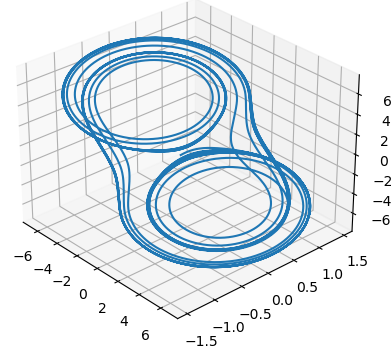}
	\caption{Phase portrait of \eqref{eq:chua} with initial conditions $(0.4, -0.2, 0.2)$}
	\label{fig:chua}
\end{figure}

Suppose a coupling strength $c=0.05$ and inner coupling matrix $\Gamma = I_3$ the $3 \times 3$ identity matrix.
Observing that 
\begin{equation*}
    \|f(t,x) - f(t,y)\| 
    \le \|B(x-y) + p|m_0 - m_1|\text{diag}\{1,0,0\}|x_1-y_1|\|
\end{equation*}
and 
\begin{equation*}
 (x-y)^T(f(t,x) - f(t,y)) \le (x-y)^T(B + p|m_0 - m_1|\text{diag}\{1,0,0\})(x-y),
\end{equation*}
we know that Assumptions \ref{as:lipschitz} and \ref{as:matrix} are satisfied by setting $K = (B + p|m_0 - m_1|\text{diag}\{1,0,0\}) \Gamma^{-1}$ and $L = \|B\| + p|m_0-m_1|$.
We choose initial conditions for each component of each $x_i$ between zero and 0.5, setting each $\alpha_i = 1$, $\mu = 0.95$, $\kappa = 5$, and $\beta = 0.2$.  
Applying the topology-based pinning strategy described above, we find that the system can be stabilized with an admissible delay of approximately 0.0042 by pinning nodes 1, 2, and 3. Then the conditions of Theorem \ref{thm:main} are satisfied by setting the control gains equal to 0.94.  
However the conditions of Theorem 1 are also satisfied by taking an actuation delay of 0.002 and control gains equal to 0.45 which we choose to create more illustrative figures.

By solving numerically, as shown in Figure \ref{fig:2ms-delay} (\cite{jitcxde}), we see that all nodes in the network are stabilized towards the origin.
The actuation delay between when the Lyapunov function for a controlled node exceeds the triggering threshold and when the control impulse occurs can be seen in this figure. 
As predicted by the proposed event-triggering mechanism, the Lyapunov function first reaches the prescribed triggering threshold, which generates a triggering event, and the corresponding impulsive update is executed after the actuation delay.
Figure~\ref{fig:2ms-delay-zoom} shows that all inter-event times over the time interval $[0,2]$ are greater than the theoretical lower bound of approximately 0.00205 derived in Theorem~\ref{thm:main}. This agreement between the theoretical analysis and the numerical results confirms the validity of the proposed triggering mechanism and verifies that the closed-loop system is free of Zeno behavior.
\begin{figure}[htpb]
    \centering
    \begin{subfigure}{0.48\textwidth}
    \includegraphics[width=\textwidth,trim=0cm 0cm 0cm 2.2cm, clip]{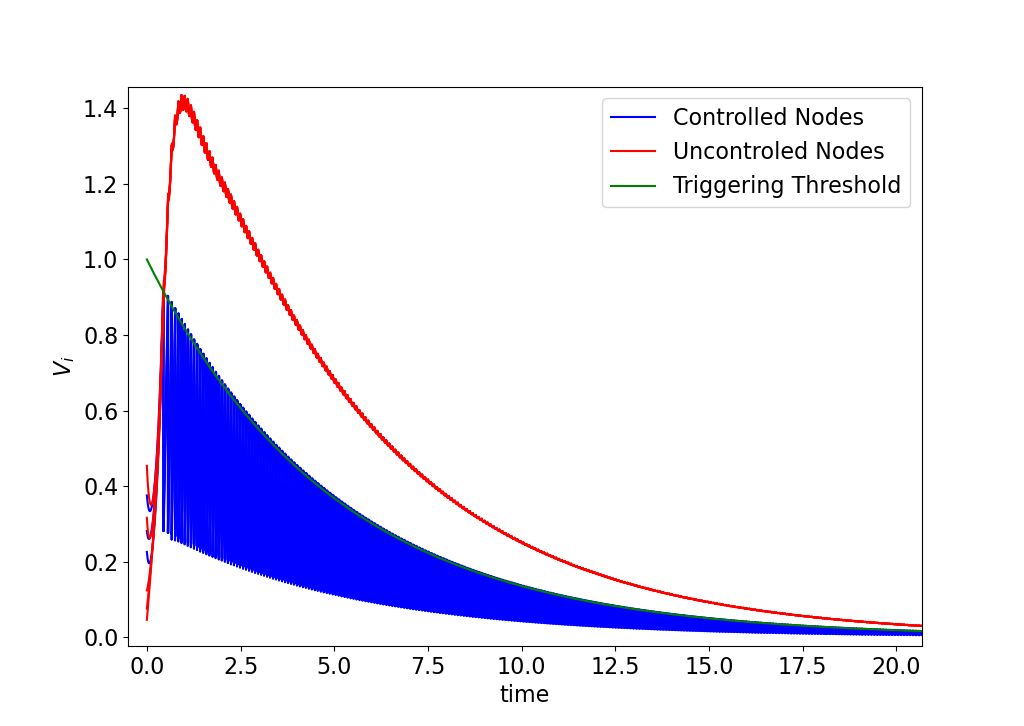}
    \caption{}\label{fig:2ms-delay}
\end{subfigure}
\hfill
\begin{subfigure}{0.48\textwidth}
    \includegraphics[width=\textwidth,trim=0cm 0cm 0cm 2.2cm, clip]{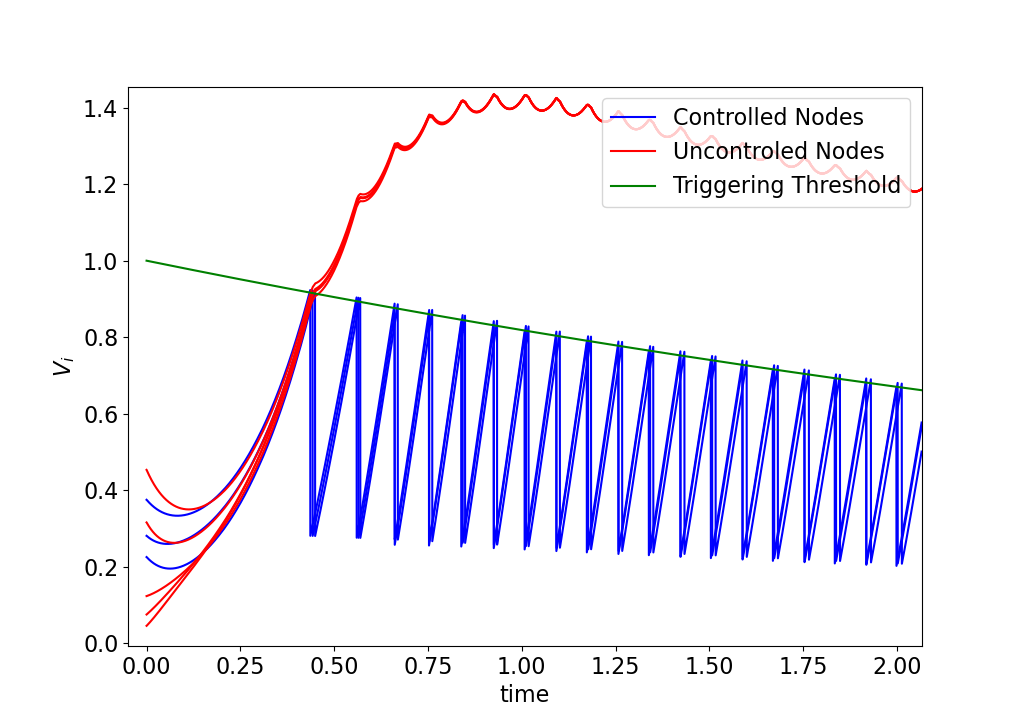}
    \caption{}\label{fig:2ms-delay-zoom}
\end{subfigure}
    \caption{Numerical simulations of a network of Chua circuits with an actuation delay of $\tau = 2\,\mathrm{ms}$. (a) Closed-loop response under the proposed event-triggered impulsive pinning control. The Lyapunov function $V_i$ of each node is shown in blue for pinned nodes and in red for unpinned nodes, while the triggering threshold is shown in green. (b) Enlarged view of panel~(a) over the interval $t \in [0,2]$, highlighting the impulsive state updates of the pinned nodes. Note that after $t=0.5$ only the pinned nodes intersect the triggering threshold, whereas the unpinned nodes remain above it throughout the simulation.} 
\end{figure}

When a larger actuation delay is considered, we can no longer guarantee that the system will be stabilized. 
Increasing the delay from 0.002 to 0.055 while keeping the other simulation parameters the same, we see in Figure \ref{fig:55ms-delay} that the network is not asymptotically stable. 
Specifically, the Lyapunov function for the controlled nodes grows enough during the delay that the impulse does not reduce it back below the triggering threshold and no further impulses are triggered for that node. 

We could also choose larger initial conditions for the unpinned nodes. 
For example, setting $x_i(t_0) = (0.5,0.5,0.5)$ for $i=1, \dots, 3$ and $x_j(t_0) = (4.5,4.5,4.5)$ for $j=4, \dots, 8$ and leaving all other parameters the same we still satisfy the condition that $W(t_0) \le \eta_0$ and $V_i(t_0) < \mu \alpha_i$. 
However, choosing an actuation delay of 0.010, we see in Figure \ref{fig:10ms-delay-large-IC} that the Lyapunov functions for the controlled nodes grow enough during their first delay that the impulse does not reduce them back below the triggering threshold so that no further impulses are triggered and the system is not stabilized. 

{The numerical simulations indicate that the theoretical delay bound provided by Theorem~\ref{thm:main} may be conservative. Therefore, the delay bound derived in this paper should be interpreted as a rigorous analytical guarantee rather than an estimate of the true maximum admissible delay.} % bar

\begin{figure}[htpb]
    \centering
    \begin{subfigure}{0.48\textwidth}
    \includegraphics[width=\textwidth,trim=0cm 0cm 0cm 2.2cm, clip]{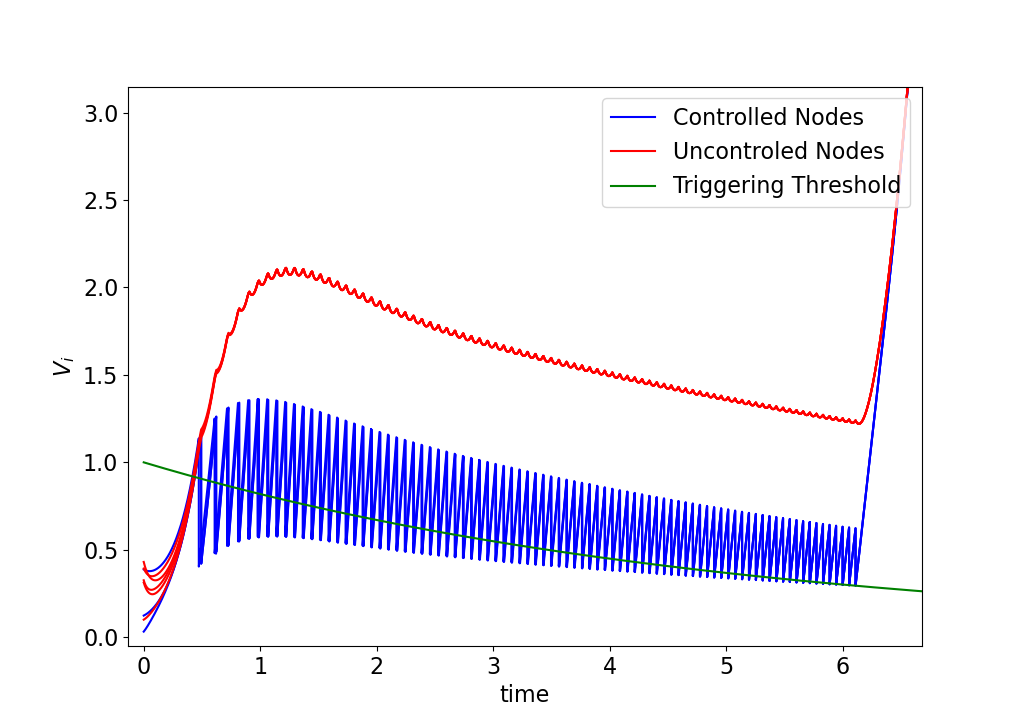}
    \caption{}\label{fig:55ms-delay}
\end{subfigure}
\hfill
\begin{subfigure}{0.48\textwidth}
    \includegraphics[width=\textwidth,trim=0cm 0cm 0cm 2.2cm, clip]{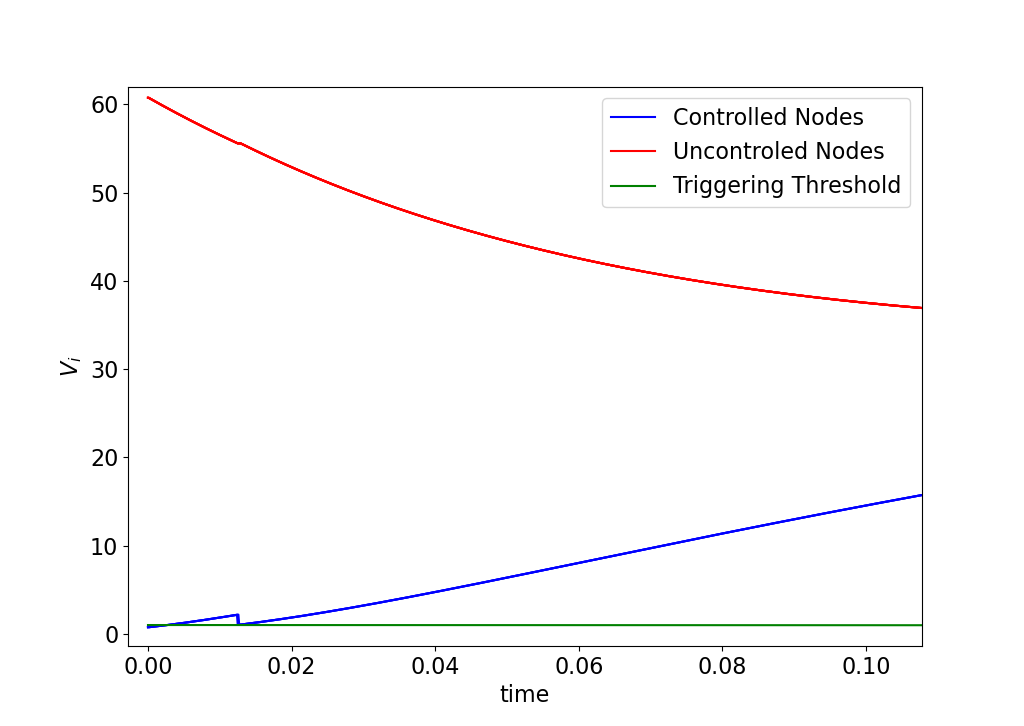}
    \caption{}\label{fig:10ms-delay-large-IC}
\end{subfigure}
    \caption{Numerical simulations of a network of Chua circuits with larger actuation delays. The Lyapunov function $V_i$ of each node is shown in blue for pinned nodes and in red for unpinned nodes, while the triggering threshold is shown in green. (a) Closed-loop response under the proposed event-triggered impulsive pinning control with an actuation delay of $\tau = 55 \mathrm{ms}$. The larger actuation delay results in a loss of stability, and the network fails to converge to the desired equilibrium. Although the pinned nodes continue to satisfy the triggering condition until around the time $t=6.2$, the delayed control actions are no longer sufficient to stabilize the network. (b) Closed-loop response  with an actuation delay of $\tau = 10 \mathrm{ms}$ and larger initial data for the unpinned nodes. The pinned nodes fail to satisfy the triggering condition after one impulse because the delay is too great, so the network is not stabilized. }
\end{figure}

\section{Conclusion}\label{Sec5}

This paper studied the stabilization of complex dynamical networks using event-triggered pinning impulsive control in the presence of actuation delays. Unlike existing works that assume instantaneous control implementation, we explicitly considered delays between event triggering and control execution, which better reflects practical systems. A Lyapunov-based framework was developed to analyze the system behavior under actuation delays. We derived explicit delay-dependent sufficient conditions ensuring asymptotic stability, and established a positive lower bound on inter-event times to exclude Zeno behavior. Numerical results demonstrated the effectiveness of the proposed method.

Future work may consider extending the proposed framework to more general settings, such as networks with heterogeneous node dynamics, stochastic disturbances, or time-varying delays. Another interesting direction is the development of distributed or adaptive event-triggered mechanisms that further reduce communication and computation requirements.

\bibliography{ifac-paper}             % bib file to produce the bibliography with bibtex (preferred)
\end{document}